\documentclass[runningheads,a4paper]{llncs}

\usepackage{amsfonts, amsmath}
\usepackage{graphicx}

\newcommand{\Geqt}{\ensuremath{G_\text{eqt}}}
\newcommand{\id}{\ensuremath{\text{id}}}
\newcommand{\seg}{\ensuremath{\text{seg}}}
\newcommand{\pred}{\ensuremath{\text{pred}}}
\renewcommand{\succ}{\ensuremath{\text{succ}}}

\title{Improved Leader Election for Self-Organizing Programmable Matter}
\titlerunning{Improved Leader Election for Programmable Matter}
\author{Joshua J. Daymude\inst{1} \and
        Robert Gmyr\inst{2} \and
        Andr\'ea W.\ Richa\inst{1} \and
        Christian Scheideler\inst{2} \and
        Thim Strothmann\inst{2}}
\authorrunning{Daymude et al.}
\institute{Computer Science, CIDSE, Arizona State University, USA\\
           \email{\{jdaymude,aricha\}@asu.edu}
           \and
           Department of Computer Science, Paderborn University, Germany\\
           \email{\{gmyr,scheidel,thim\}@mail.upb.de}}

\begin{document}

\maketitle

\begin{abstract}
We consider programmable matter that consists of computationally limited devices (called {\em particles}) that are able to self-organize in order to achieve some collective goal without the need for central control or external intervention.
We use the geometric amoebot model to describe such self-organizing particle systems, which defines how particles can actively move and communicate with one another.
In this paper, we present an efficient local-control algorithm which solves the leader election problem in $\mathcal{O}(n)$ asynchronous rounds with high probability, where $n$ is the number of particles in the system.
Our algorithm relies only on local information --- particles do not have unique identifiers, any knowledge of $n$, or any sort of global coordinate system --- and requires only constant memory per particle.
\end{abstract}

\section{Introduction} \label{sec:intro}

The vision for {\em programmable matter} is to create some material or substance that can change its physical properties like shape, density, conductivity, or color in a programmable fashion based on either user input or autonomous sensing of its environment.
Many realizations of programmable matter have been proposed --- including DNA tiles, shape-changing molecules, synthetic cells, and reconfiguring modular robots --- each of which is pursuing solutions applicable to its own situation, subject to domain-specific capabilities and constraints.
We envision programmable matter as a more abstract system of computationally limited devices (which we refer to as \emph{particles}) which can move, bond, and exchange information in order to collectively reach a given goal without any outside intervention.
\emph{Leader election} is a central and classical problem in distributed computing that is very interesting for programmable matter; e.g., most known shape formation techniques for programmable matter suppose the existence of a leader/seed particle (examples can be found in \cite{winfree13} for the nubot model, \cite{RothemundW00} for the abstract tile self assembly model and \cite{Nanocom,SPAA16} for the amoebot model).

In this paper, we present a fully asynchronous local-control protocol for the leader election problem, improving our previous algorithm for leader election in~\cite{DNA} which was only described at a high level, lacking specific rules for each particle's execution. Moreover, while the analysis in~\cite{DNA} used a simplified, synchronous setting and only achieved its linear runtime bound in expectation, here we prove with high probability\footnote{An event occurs \emph{with high probability} (\emph{w.h.p.}), if the probability of success is at least $1 - n^{-c}$, where $c > 1$ is a constant; in our context, $n$ is the number of particles.} correctness\footnote{An astute reader may note that a w.h.p.~guarantee on correctness is weaker than the absolute guarantee given for the algorithm in~\cite{DNA}, but the latter was given without considering the necessary particle-level execution details.} and runtime guarantees for the full local-control protocol. Finally, as this algorithm is both conceptually simpler than that of~\cite{DNA} and intended for each particle to run individually, it is more easily understood and implemented.

\subsection{Amoebot Model} \label{subsec:model}

We represent any structure the particle system can form as a subgraph of the infinite graph $G = (V,E)$, where $V$ represents all possible positions the particles can occupy relative to their structure, and $E$ represents all possible atomic movements a particle can perform
as well as all places where neighboring particles can bond to each other.
In the {\em geometric amoebot model}, we assume that $G = \Geqt$, where $\Geqt$ is the infinite regular triangular grid graph (see Figure~\ref{fig:model}a). We recall the properties of the geometric amoebot model necessary for this algorithm; a full description can be found in~\cite{DNA}.

\begin{figure}[h]
\centering
\includegraphics[width = 0.95\columnwidth]{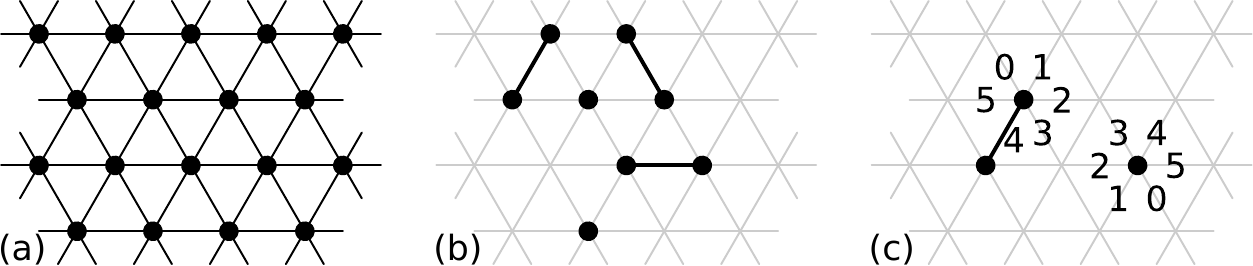}
\caption{(a) A section of $\Geqt$, with black circles depicting $V$. (b) Two contracted particles, depicted as black circles, and three expanded particles, depicted as black circles connected by an edge, on $\Geqt$, depicted as a gray mesh. (c) Two particles occupying non-adjacent positions in $\Geqt$ with different offsets for their head port labellings.}
\label{fig:model}
\end{figure}

Each particle occupies either a single node (i.e., it is \emph{contracted}) or a pair of adjacent nodes in $\Geqt$ (i.e., it is \emph{expanded}), and every node can be occupied by at most one particle (see Figure~\ref{fig:model}b).
Particles move through \emph{expansions} and \emph{contractions}; however, as our leader election algorithm does not require particles to move, we omit a detailed description of these movement mechanisms.

Particles are \emph{anonymous}; they have no unique identifiers.
Instead, each particle has a collection of {\em ports} --- one for each edge incident to the node(s) the particle occupies --- that have unique labels from the particle's local perspective.
We assume that the particles have a common {\em chirality} (i.e., a shared notion of clockwise direction), which allows each particle to label its ports in clockwise order.
However, particles do not have a common sense of global orientation and may have different offsets for their port labels (see, e.g., Figure~\ref{fig:model}c).

Two particles occupying adjacent nodes are connected by a \emph{bond}, and we refer to such particles as \emph{neighbors}. Neighboring particles establish bonds via the ports facing each other. The bonds not only ensure that the particle system forms a connected structure, but also are used for exchanging information.
Each particle has a constant-size local memory that can be read and written to by any neighboring particle. 
Particles exchange information with their neighbors by simply writing into their memory.
A particle always knows whether it is contracted or expanded, and we assume that this information is also available to its neighbors (by publishing it in its local memory).
Due to the constant-size memory constraint, particles know neither the total number of particles in the system nor any estimate of this number.

We assume the standard asynchronous model, wherein particles execute an algorithm concurrently and no assumptions are made about individual particles' activation rates or computation speeds.
A classical result under this model is that for any asynchronous concurrent execution of atomic particle activations, there exists a sequential ordering of the activations which produces the same end configuration, provided conflicts which arise from the concurrent execution are resolved.
Thus, it suffices to view particle system progress as a sequence of \emph{particle activations}; i.e., only one particle is active at a time.
Whenever a particle is activated, it can perform an arbitrary, bounded amount of computation involving its local memory and the memories of its neighbors and can perform at most one movement. We define an \emph{asynchronous round} to be complete once each particle has been activated at least once.

\subsection{Related Work} \label{subsec:relwork}

A variety of work related to programmable matter has recently been proposed and investigated.
One can distinguish between active and passive systems. In passive systems, the computational units either have no intelligence (moving and bonding is based only on their structural properties or interactions with their environment), or have limited computational capabilities but cannot control their movements.
Examples of research on \emph{passive systems} are DNA computing~\cite{Adl94,BDLS96,CDBG11,WLWS98}, tile self-assembly systems (e.g., the surveys in~\cite{doty2012,patitz2014,Woods2013intrinsic}), and population protocols~\cite{AAD+06}.
We will not describe these models in detail as they are of little relevance to our approach.
\emph{Active systems}, on the other hand, are composed of computational units which can control the way they act and move in order to solve a specific task. We discuss prominent examples of active systems here, as they are more comparable to our work.

In the area of \emph{swarm robotics}, it is usually assumed that there is a collection of autonomous robots that can move freely in a given area and have limited sensing, vision, and communication ranges.
They are used in a variety of contexts, including graph exploration (e.g.,~\cite{fl13}), gathering problems (e.g.,~\cite{AG3,ci12}), and shape formation problems (e.g.,~\cite{fl08,kilobots}).
Surveys of recent results in swarm robotics can be found in~\cite{Ker12,McL08}; other samples of representative work can be found in~\cite{AR10,BFMS11,CP08,DFSY10,DS08,HABFM02,KM11}.
While the analytic techniques developed in swarm robotics and natural swarms are of some relevance to this work, the individual units in those systems have more powerful communication and processing capabilities than in the systems we consider.

The field of \emph{modular self-reconfigurable robotic systems} focuses on intra-robotic aspects such as design, fabrication, motion planning, and control of autonomous kinematic machines with variable morphology (e.g.,~\cite{FNKB88,YSS+07}).
\emph{Metamorphic robots} form a subclass of self-reconfigurable robots that share some of the characteristics of our geometric amoebot model~\cite{Chi94}.
Hardware development in the field of self-reconfigurable robotics has been complemented by a number of algorithmic advances (e.g.,~\cite{BKRT04,kilobots,WWA04}), but mechanisms that automatically scale from a few to hundreds or thousands of individual units are still under investigation, and no rigorous theoretical foundation is available yet.

The \emph{nubot} model~\cite{chen2014fast,chen2013parallel,winfree13} by Woods et al.~aims to provide the theoretical framework that would allow for more rigorous algorithmic studies of biomolecular-inspired systems, specifically of self-assembly systems with active molecular components.
While there are similarities between such systems and our self-organizing particle systems, key differences prohibit the translation of the algorithms and other results under the nubot model to our systems; e.g., there is always an arbitrarily large supply of ``extra'' particles that can be added to the nubot system as needed, and a (non-local) notion of rigid-body movement.

The \emph{amoebot} model~\cite{spaa-ba14} is a model for self-organizing programmable matter that aims to provide a framework for rigorous algorithmic research for nano-scale systems.
In~\cite{DNA}, the authors describe a leader election algorithm for an abstract (synchronous) version of the amoebot model that decides the problem in expected linear time.
Recently, a universal shape formation algorithm~\cite{SPAA16}, a universal coating algorithm~\cite{DNA16} and a Markov chain algorithm for the compression problem~\cite{PODCSarah} were introduced, showing that there is potential to investigate a wide variety of problems under this model.

\subsection{Problem Description} \label{subsec:problem}

We consider the classical problem of \emph{leader election}.
An algorithm is said to solve the leader election problem if for any connected particle system of initially contracted particles with empty memories, eventually a single particle \emph{irreversibly} declares itself the \emph{leader} (e.g., by setting a dedicated bit in its memory) and no other particle ever declares itself to be the leader.
We define the running time of a leader election algorithm to be the number of asynchronous rounds until a leader is declared.
Note that we do not require the algorithm to terminate for particles other than the leader. We investigate several variants of the leader election problem in Section~\ref{app:variants}, such as allowing initial configurations which contain expanded particles or requiring the algorithm to terminate for all particles.

\subsection{Our Contributions} \label{subsec:contribs}

In this paper, we introduce a randomized leader election algorithm for programmable matter that requires $\mathcal{O}(n)$ asynchronous rounds with high probability, where $n$ is the number of particles in the system.
We present our algorithm in Section~\ref{sec:algo} and analyze its correctness and runtime in Section~\ref{sec:analysis}.
In Section~\ref{app:variants}, we show how to adapt the algorithm such that $(i)$ all non-leader particles are in a specific non-leader state once the leader has been elected and $(ii)$ we achieve a success probability of $1$.

This new leader election algorithm has the following advantages over our algorithm in~\cite{DNA}: $(i)$ we describe a full local-control protocol as opposed to merely sketching how a high level protocol could be executed by individual particles, $(ii)$ correctness and runtime guarantees are proven for the local-control protocol as opposed to doing so only for a simplified synchronous version of the algorithm as in~\cite{DNA}, and $(iii)$ the asynchronous runtime guarantees hold with high probability as opposed to just in expectation, which was the case for the aforementioned synchronous version.

\section{Algorithm} \label{sec:algo}

Before we describe the leader election algorithm in detail, we give a short high-level overview.
The algorithm consists of six \emph{phases}.
These phases are not strictly synchronized among each other, i.e., at any point in time, different parts of the particle system may execute different phases.
Furthermore, a particle can be involved in the execution of multiple phases at the same time.
The first phase is \emph{boundary setup} (Section~\ref{subsec:boundsetup}).
In this phase, each particle locally checks whether it is part of a \emph{boundary} of the particle system.
Only the particles on a boundary participate in the leader election.
Particles occupying a common boundary organize themselves into a directed cycle.
The remaining phases operate on each boundary independently.
In the \emph{segment setup} phase (Section~\ref{subsec:segsetup}), the boundaries are subdivided into \emph{segments}: each particle flips a fair coin.
Particles that flip heads become \emph{candidates} and compete for leadership whereas particles that flip tails become \emph{non-candidates} and assist the candidates in their competition.
A segment consists of a candidate and all subsequent non-candidates along the boundary up to the next candidate.
The \emph{identifier setup} phase (Section~\ref{subsec:idsetup}) assigns a random identifier to each candidate.
The identifier of a candidate is stored distributively among the particles of its segment.
In the \emph{identifier comparison} phase (Section~\ref{subsec:idcomp}), the candidates compete for leadership by comparing their identifiers using a token passing scheme.
Whenever a candidate sees an identifier that is higher than its own, it revokes its candidacy.
Whenever a candidate sees its own identifier, the \emph{solitude verification} phase (Section~\ref{subsec:solitudever}) is triggered.
In this phase, the candidate checks whether it is the last remaining candidate on the boundary.
If so, it initiates the \emph{boundary identification} phase (Section~\ref{subsec:boundid}) to determine whether it occupies the unique \emph{outer boundary} of the system.
In that case, it becomes the leader; otherwise, it revokes its candidacy.

\subsection{Boundary Setup} \label{subsec:boundsetup}

The boundary setup phase organizes the particle system into a set of \emph{boundaries}. This approach is directly adopted from~\cite{DNA}, but we give a full description here to introduce important notation.
Let $A \subset V$ be the set of nodes in $\Geqt = (V,E)$ that are occupied by particles.
According to the problem definition, the subgraph $\Geqt|_A$ of $\Geqt$ induced by $A$ is connected.
Consider the graph $\Geqt|_{V \setminus A}$ induced by the unoccupied nodes in $\Geqt$.
We call a connected component $R$ of $\Geqt|_{V \setminus A}$ an \emph{empty region}.
Let $N(R)$ be the neighborhood of an empty region $R$ in $\Geqt$; that is, $N(R) = \{u \in V \setminus R : \exists v \in R \text{ such that } \{u,v\} \in E\}$.
Note that by definition, all nodes in $N(R)$ are occupied by particles.
We refer to $N(R)$ as the \emph{boundary} of the particle system corresponding to $R$.
Since $\Geqt|_A$ is a finite graph, exactly one empty region has infinite size while the remaining empty regions have finite size.
We define the boundary corresponding to the infinite empty region to be the unique \emph{outer boundary} and refer to a boundary that corresponds to a finite empty region as an \emph{inner boundary}.

For each boundary of the particle system, we organize the particles occupying that boundary into a directed cycle.
Upon first activation, each particle instantly determines its place in these cycles using only local information.
Figure~\ref{fig:boundsetup} shows all possible neighborhoods of a particle (up to rotation) and the corresponding results of the boundary setup phase.
To produce the depicted results, a particle $p$ proceeds as follows.
First, $p$ checks for two special cases. If $p$ has no neighbors, it must be the only particle in the particle system since the particle system is connected. Thus, it immediately declares itself the leader and terminates. If all neighboring nodes of $p$ are occupied, $p$ is not part of any boundary and terminates without participating in the leader election process any further.

\begin{figure}[tbh]
\centering
\includegraphics[scale=0.6]{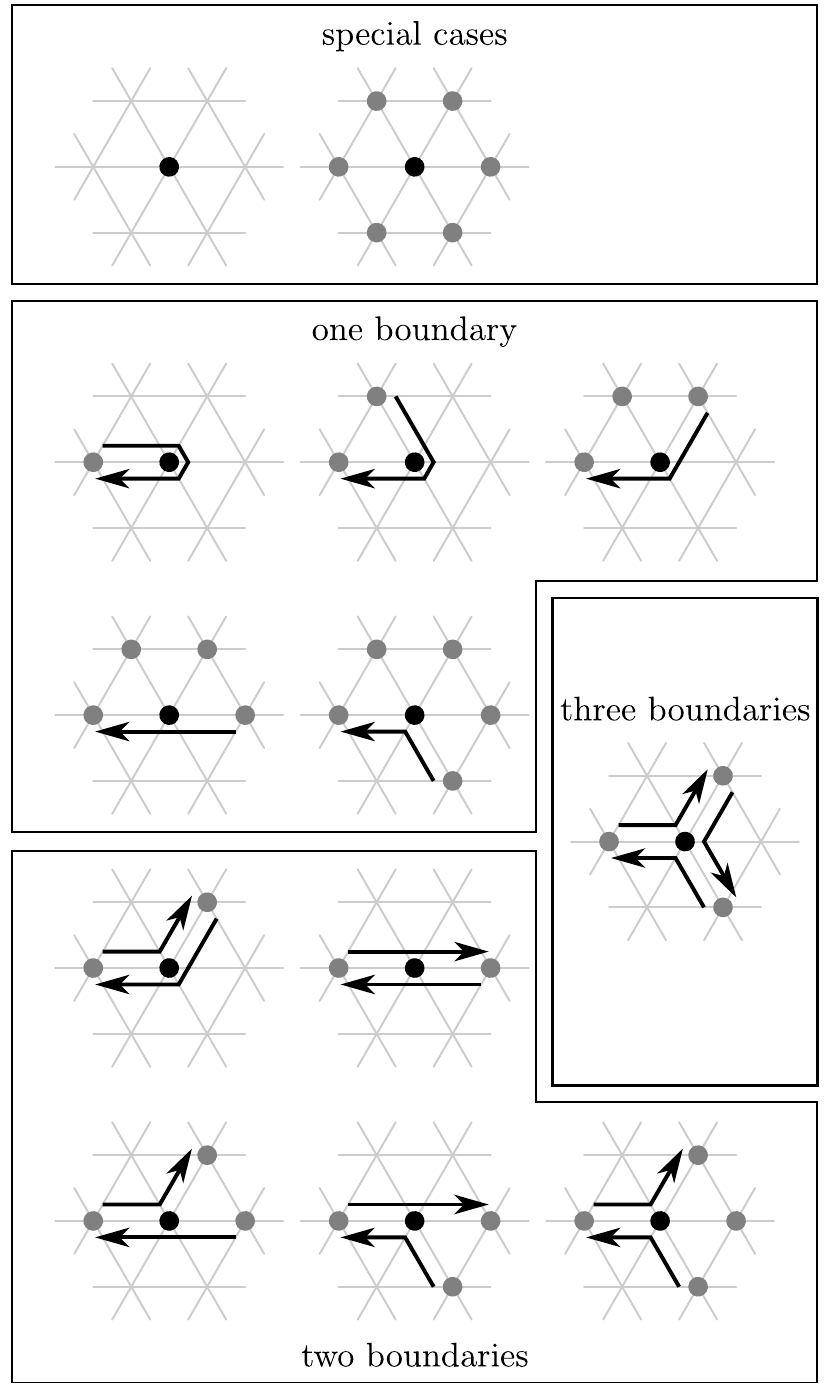}
\caption{Possible results (up to rotation) of the boundary setup phase depending on the neighborhood of a particle. In the special cases shown in the top-most part of the figure the particle does not consider itself part of any boundary. The remainder of the figure is organized according to the number of boundaries a particle is part of from its local perspective. For each boundary, the figure shows an arrow that starts at the predecessor and ends at the successor of the particle on that boundary.}
\label{fig:boundsetup}
\end{figure}

If these special cases do not apply, then $p$ has at least one occupied node and one unoccupied node in its neighborhood.
Interpret the neighborhood of $p$ as a directed ring of six nodes that is oriented clockwise around $p$.
Consider all maximal sequences of unoccupied nodes $(v_1, \ldots v_k)$ in this ring; call such a sequence an \emph{empty sequence}.
Such a sequence is part of some empty region and hence corresponds to a boundary that includes $p$.
Let $v_0$ be the node before $v_1$ and let $v_{k+1}$ be the node after $v_k$ in the ring.
Note that we might have $v_0 = v_{k+1}$.
By definition, $v_0$ and $v_{k+1}$ are occupied.
Particle $p$ implicitly arranges itself as part of a directed cycle spanning the aforementioned boundary by considering the particle occupying $v_0$ to be its \emph{predecessor} and the particle occupying $v_{k+1}$ to be its \emph{successor} on that boundary.
It repeats this process for each empty sequence in its neighborhood.

A particle can have up to three empty sequences in its neighborhood (see Figure~\ref{fig:boundsetup}), and consequently can be part of up to three distinct boundaries.
However, a particle cannot locally decide whether two distinct empty sequences belong to two distinct empty regions or to the same empty region.
To guarantee that the executions on distinct boundaries are isolated, we let the particles treat each empty sequence as a distinct empty region.
For each such sequence, a particle acts as a distinct \emph{agent} which executes an independent instance of the algorithm encompassing the remaining five phases of the leader election algorithm.
Whenever a particle is activated, it sequentially executes the independent instances of the algorithm for each of its agents in an arbitrary order, i.e., whenever a particle is activated also its agents are activated.
Each agent $a$ is assigned the predecessor and successor --- denoted $a.\pred$ and $a.\succ$, respectively --- that was determined by the particle for its corresponding empty sequence.
This organizes the set of all agents into disjoint cycles spanning the boundaries of the particle system (see Figure~\ref{fig:agents}).
As consequence of this approach, a particle can occur up to three times on the same boundary as different agents.
While we can ignore this property for most of the remaining phases, it will remain a cause for special consideration in the solitude verification phase (Section~\ref{subsec:solitudever}).

\begin{figure}
\centering
\includegraphics{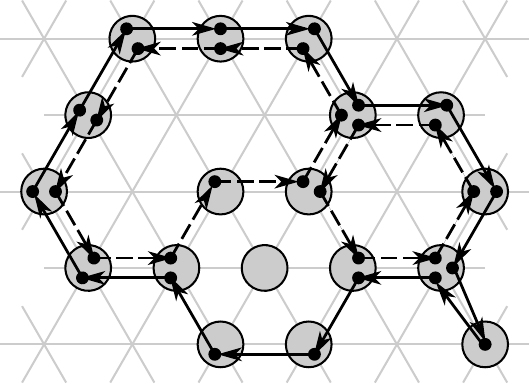}
\caption{Boundaries and agents. Particles are depicted as gray circles and the agents of a particle are depicted as black dots inside of the corresponding circle. After the boundary setup phase, the agents form disjoint cycles that span the boundaries of the particle system. The solid arrows represent the unique outer boundary and the dashed arrows represent the two inner boundaries.}
\label{fig:agents}
\end{figure}

\subsection{Segment Setup} \label{subsec:segsetup}

All remaining phases (including this one) operate exclusively on boundaries, and furthermore execute on each boundary independently.
Therefore, we only consider a single boundary for the remainder of the algorithm description.
The goal of the segment setup phase is to divide the boundary into disjoint ``segments''.
Each agent flips a fair coin.
The agents which flip heads become \emph{candidates} and the agents which flip tails become \emph{non-candidates}.
In the following phases, candidates compete for leadership while non-candidates assist the candidates in their competition.
A \emph{segment} is a maximal sequence of agents $(a_1, a_2, \ldots, a_k)$ such that $a_1$ is a candidate, $a_i$ is a non-candidate for $i > 1$, and $a_i = a_{i-1}.\succ$ for $i > 1$.
Note that the maximality condition implies that the successor of $a_k$ is a candidate.
We refer to the segment starting at a candidate $c$ as $c.\seg$ and call it the segment of $c$.
In the following phases, each candidate uses its segment as a distributed memory.

\subsection{Identifier Setup} \label{subsec:idsetup}

After the segments have been set up, each candidate generates a random \emph{identifier} by assigning a random digit to each agent in its segment.
The candidates use these identifiers in the next phase to engage in a competition in which all but one candidate on the boundary are eliminated.
Note that the term identifier is slightly misleading in that two distinct candidates can have the same identifier.
Nevertheless, we hope that the reader agrees that the way these values are used makes this term an appropriate choice.

To generate a random identifier, a candidate $c$ sends a \emph{token} along its segment in the direction of the boundary.
A token is simply a constant-size piece of information that is passed from one agent to the next by writing it to the memory of a neighboring particle.
Throughout the leader election algorithm, a particle holds at most a constant number of tokens at all times, so the constant-size memory of the particles is sufficient to implement a token-passing mechanism.
While the token traverses the segment, it assigns a value chosen uniformly at random from $[0, r - 1]$ to each visited agent where $r$ is a constant that is fixed in the analysis.
The identifier generated in this way is a number with radix $r$ consisting of $|c.\seg|$ digits where $c$ holds the most significant digit and the last agent of $c.\seg$ holds the least significant digit.
We refer to the identifier of a candidate $c$ as $c.\id$.
The competition in the next phase of the algorithm is based on comparing identifiers.
When comparing identifiers of different lengths, we define the shorter identifer to be lower than the longer identifier.

After generating its random identifier, each candidate creates a copy of its identifier that is stored in \emph{reversed digit order} in its segment.
This step is required as a preparation for the next phase.
To achieve this, we use a single token that moves back and forth along the segment and copies one digit at a time.
More specifically, we reuse the token described above that generated the random identifier.
Once this token reaches the end of the segment, it starts copying the identifier by reading the digit of the last agent of the segment and moving to the beginning of the segment.
There, it stores a copy of that digit in the candidate $c$.
It then reads the digit of $c$ and moves back to the end of the segment where it stores a copy of that digit in the last agent of the segment.
It proceeds in a similar way with the second and the second to last agent and so on until the identifier is completely copied.
Afterwards, the token moves back to $c$ to inform the candidate that the identifier setup is complete.

Note that for ease of presentation we deliberately opted for simplicity over speed when creating a reversed copy of the identifier.
As we will show in Section~\ref{subsec:runningtime}, the running time of this simple algorithm is dominated by the running time of the next phase so that the overall asymptotic running time of the leader election algorithm does not suffer.

\subsection{Identifier Comparison} \label{subsec:idcomp}

During the identifier comparison phase the agents use their identifiers to compete with each other.
Each candidate compares its own identifier with the identifier of every other candidate on the boundary.
A candidate with the highest identifier eventually progresses to the solitude verification phase, described in the next section, while any candidate with a lower identifier withdraws its candidacy.
To achieve the comparison, the non-reversed copies of the identifiers remain stored in their respective segments while the reversed copies move backwards along the boundary as a sequence of tokens.
More specifically, a \emph{digit token} is created for each digit of a reversed identifier.
A digit token created by the last agent of a segment is marked as a \emph{delimiter token}.
Once created, the digit tokens traverse the boundary against the direction of the cycle spanning it.
Each agent is allowed to hold at most two tokens at a time, which gives the tokens some space to move along the boundary.
The tokens are not allowed to overtake each other, so whenever an agent stores two tokens, it keeps track of the order they were received in and forwards them accordingly.
An agent forwards at most one token per activation.
Furthermore, an agent can only receive a token after it creates its own digit token.
We define the \emph{token sequence} of a candidate $c$ as the sequence of digit tokens created by the agents in $c.\seg$.
Note that according to the rules for forwarding tokens, the token sequences of distinct candidates remain separated and the tokens within a token sequence maintain their relative order along the boundary.

Whenever a token sequence traverses a segment $c.\seg$ of a candidate $c$, the agents in $c.\seg$ cooperate with the tokens of the token sequence to compare the identifier $c.\id$ with the identifier stored in the token sequence.
This comparison has three possible outcomes: $(i)$ the token sequence is longer than $c.\seg$ or the lengths are equal and the token sequence stores an identifier that is strictly greater than $c.\id$, $(ii)$ the token sequence is shorter than $c.\seg$ or the lengths are equal and the token sequence stores an identifier that is strictly smaller than $c.\id$, or $(iii)$ the lengths are equal and the identifiers are equal.
In the first case, $c$ does not have the highest identifier and withdraws its candidacy.
In the second case, $c$ might be a candidate with the highest identifier and therefore remains a candidate.
Finally, in the third case $c$ initiates the solitude verification phase, which is then executed in parallel to the identifier comparison phase.
Solitude verification might be triggered quite frequently, especially for candidates with short segments; we describe how this is handled in the next section.

To describe the token passing scheme for identifier comparison, consider a candidate $c$ and the next candidate $c'$ along the boundary at the beginning of the identifier comparison phase. As token sequences move backwards along the boundary, the token sequence of $c'$ will be compared with $c.\id$.
Both agents and tokens can be either \emph{active} or \emph{inactive}.
Agents are initially active while tokens are initially inactive.
When an active agent receives an active token, both become inactive and we say the agent and the token \emph{match}.
Since the tokens in the token sequence of $c'$ are initially inactive, they are forwarded by the agents of $c'.\seg$ without matching.
Whenever a token is forwarded by a candidate into a new segment, the token becomes active.
Therefore, the tokens in the token sequence of $c'$ are active when they enter $c.\seg$.
When an agent matches with a token, it compares its digit of the \emph{non-reversed} identifier with the digit stored in the token and keeps the result of the comparison for future reference.
Note that since the digits of $c'.\id$ are stored in \emph{reversed order} in the token sequence of $c'$, the agent holding the least significant digit of $c.\id$ matches with the token holding the least significant digit of $c'.\id$, the agent holding the second to least significant digit of $c.\id$ matches with the token holding the second to least significant digit of $c'.\id$, and so on.

The lengths of the identifiers are compared as follows.
Recall that the digit token generated by the last agent of a segment is marked as a delimiter token.
If the delimiter token of $c'$ matches with $c$, the token sequence of $c'$ has the same length as $c.\seg$.
If the delimiter token of $c'$ matches with another agent of $c.\seg$ (and is therefore already inactive when it reaches $c$), the token sequence of $c'$ is shorter than $c.\seg$.
Finally, if a non-delimiter token of $c'$ matches with $c$, the token sequence of $c'$ is longer than $c.\seg$.
Consequently, $c$ can distinguish these three cases once it receives the delimiter token of $c'$.

If the token sequence of $c'$ has the same length as $c.\seg$, $c$ has to compare its identifier with the identifier stored in the token sequence of $c'$.
As described above, the digits of the respective identifiers are compared in the correct order and the results of the comparisons are stored distributively by the agents of $c.\seg$.
When the delimiter token of $c'$ traverses $c.\seg$, it keeps track of the comparison result of the most significant digit for which the identifiers differ.
It can do so because during its traversal it sees the consecutive digit-wise comparisons going from the least significant digit to the most significant digit.
Once $c$ receives the delimiter token of $c'$, $c$ can use the information stored within the token to decide whether the identifiers are equal or, if not, which identifier is greater.

It remains to describe how the agents and tokens are prepared for subsequent comparisons.
Specifically, we have to define when inactive agents and tokens become active again and when the comparison results stored in the agents are deleted.
As described above, an inactive token becomes active when it is forwarded by a candidate into a new segment.
The remaining tasks are the responsibility of the delimiter tokens: when an agent receives a delimiter token, it executes the computations described above and then deletes its comparison result and becomes active again.

Finally, note that a candidate that withdrew its candidacy still takes part in the identifier comparison to a certain extent.
Since the agents in its segment still match with incoming tokens, the candidate has to keep activating these tokens when it forwards them to the segment of the preceding candidate.
However, candidates that withdrew their candidacy will never progress to solitude verification and are treated as non-candidates in the solitude verification phase.

\subsection{Solitude Verification} \label{subsec:solitudever}

The goal of the solitude verification phase is for a candidate $c$ to check whether it is the last remaining candidate on its boundary.
Solitude verification is triggered during the identifier comparison phase whenever a candidate detects equality between its own identifier and the identifier of a token sequence that traversed its segment.
Note that such a token sequence can either be the token sequence created by $c$ itself or the token sequence created by some other candidate that generated the same identifier.
Once the solitude verification phase is started, it runs in parallel to the identifier comparison phase and does not interfere with it.

A candidate $c$ can check whether it is the last remaining candidate on its boundary by determining whether or not the next candidate in direction of the cycle is $c$ itself.
To achieve this, the solitude verification phase has to span not only $c.\seg$ but also all subsequent segments of former candidates that already withdrew their candidacy during the identifier comparison phase.
We refer to the union of these segments as the \emph{extended segment} of $c$.
The basic idea of the algorithm is the following.
We treat the edges that connect the agents on the boundary as vectors in the two-dimensional Euclidean plane.
If $c$ is the last remaining candidate on its boundary, the vectors corresponding to the directed edges of the boundary cycle in the extended segment of $c$ and the next edge (connecting the extended segment of $c$ to the next candidate) sum to the zero vector, implying that the next candidate and $c$ occupy the same node.
To perform this summation in a local manner, $c$ locally defines a two-dimensional coordinate system (e.g., the coordinate system depicted in Figure~\ref{fig:coordboundbox}) and uses two token passing schemes to generate and sum the $x$ and $y$ coordinates of these vectors in parallel.
We only describe the token passing scheme for the $x$-coordinates, as the scheme for the $y$-coordinates works analogously.

\begin{figure}
\centering
\includegraphics{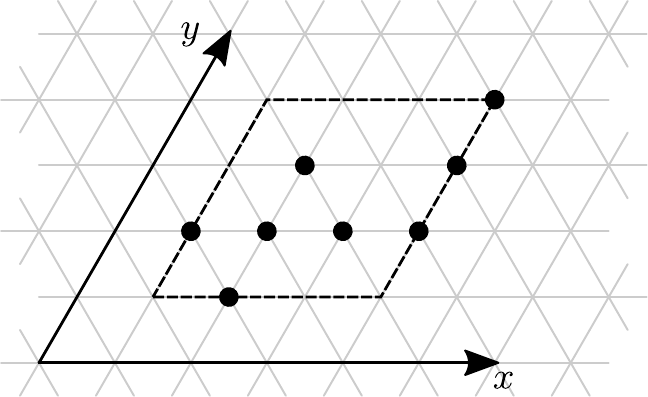}
\caption{An example of a coordinate system and bounding box of a particle system.}
\label{fig:coordboundbox}
\end{figure}

First, $c$ sends an \emph{activation token} along its extended segment to the next candidate.
Whenever the token moves \emph{right} (i.e., in \emph{positive} direction of the local $x$-axis defined by $c$), it creates a \emph{positive token} that is sent back along the boundary towards $c$.
Whenever the token moves \emph{left} (i.e., in \emph{negative} direction of the local $x$-axis), it creates a \emph{negative token} that is also sent back towards $c$.
The positive and negative tokens move independently of each other.
However, a token of either type cannot overtake another token of the same type.
Each agent can hold at most two tokens of each type.
The tokens do not move beyond $c$.

Consider a positive token $t$.
Eventually, $t$ reaches an agent $a$ such that the agents from $c$ to $a.\pred$ all hold two positive tokens.
At this point, $t$ cannot move any closer to $c$ and we say $t$ has \emph{settled}.
Each token holds a bit that specifies whether it has settled which is initially false.
It is set to true if the token reaches $c$ or if the token is held by an agent $a$ such that $a.\pred$ already holds two settled positive tokens.
The negative tokens use the same mechanism to detect whether they are settled.
Observe that once all tokens of a specific type have settled, they form a consecutive sequence whose length corresponds to the number tokens.

After the activation token completely traverses the extended segment of $c$ and reaches the next candidate, it moves back towards $c$ while staying behind the positive and negative tokens.
When it first encounters an agent $a$ in front of it that holds a settled token, it moves to $a$ and waits until all tokens at $a$ are either forwarded or have settled.
At this point, the observation from the last paragraph implies the following property:
The total number of positive tokens equals the total number of negative tokens if and only if the number of settled positive tokens at $a$ equals the number of settled negative tokens at $a$.
Therefore, the activation token can locally decide whether the $x$-coordinate of the summed vector is zero or not.
The activation token then moves back to $c$ and reports the result, deleting all positive and negative tokens it encounters along the way.
Once $c$ receives the result for both the $x$ and the $y$ coordinate, it knows whether the summed vector equals the zero vector and thus can decide whether the next candidate along the boundary is itself.

However, this is not sufficient to decide whether $c$ is the last remaining candidate on the boundary.
As described in Section~\ref{subsec:boundsetup}, a particle can occur up to three times as different agents on the same boundary.
Therefore, there can be distinct agents on the same boundary that occupy the same node of $\Geqt$.
If an extended segment reaches from one of these agents to another, the vectors induced by the extended segment sum up to the zero vector even though there are at least two agents left on the boundary.
To handle this case, each particle assigns a locally unique agent identifier from $\{1, 2, 3\}$ to each of its agents in an arbitrary way.
When the activation token reaches the end of the extended segment, it reads the agent identifier of the candidate at the end of the extended segment and carries this information back to $c$.
It is not hard to see that $c$ is the last remaining candidate on the boundary if and only if the vectors sum to the zero vector and the agent identifier stored in the activation token equals the agent identifier of $c$.

Finally, we must address the interaction between the solitude verification phase and the identifier comparison phase.
As noted in the previous section, solitude verification may be triggered quite frequently.
Therefore, it may occur that solitude verification is triggered for a candidate $c$ while $c$ is still performing a previously triggered execution of solitude verification.
In this case, $c$ simply continues with the already ongoing execution and ignores the request for another execution.
Furthermore, $c$ might be eliminated by the identifier comparison phase while it is performing solitude verification.
In this case, $c$ waits for the ongoing solitude verification to finish and only then withdraws its candidacy.

\subsection{Boundary Identification} \label{subsec:boundid}

Once a candidate $c$ determines that it is the only remaining candidate on its boundary, it initiates the boundary identification phase to check whether or not it lies on the unique outer boundary of the particle system.
If it lies on the outer boundary, the particle acting as candidate agent $c$ declares itself the leader.
Otherwise, $c$ revokes its candidacy.
To achieve this, we make use of the observation that the outer boundary is oriented clockwise while an inner boundary is oriented counter-clockwise (see Figure~\ref{fig:agents}), a property resulting directly from the way the an agent's predecessor and successor are defined in Section~\ref{subsec:boundsetup}.

A candidate $c$ can distinguish between clockwise and counter-clockwise oriented boundaries using a simple token passing scheme.
It sends a token along the boundary that sums up the angles of the turns it takes according to Figure~\ref{fig:boundid}, storing the results in a counter $\alpha$.
When the token returns to $c$, the absolute value $|\alpha|$ represents the external angle of the polygon induced by the boundary.
It is well-known that the external angle of polygon in the Euclidean plane is $|\alpha| = 360^\circ$.
Since the outer boundary is oriented clockwise and an inner boundary is oriented counter-clockwise, we have $\alpha = 360^\circ$ for the outer boundary and $\alpha = -360^\circ$ for an inner boundary.
The token can encode $\alpha$ as an integer $k$ such that $\alpha = k \cdot 60^\circ$.
To distinguish the two possible final values of $k$ it is sufficient to store $k$ modulo $5$ so that we have $k = 1$ for the outer boundary and $k = 4$ for an inner boundary.
Therefore, the token only needs three bits of memory.

\begin{figure}
\centering
\includegraphics{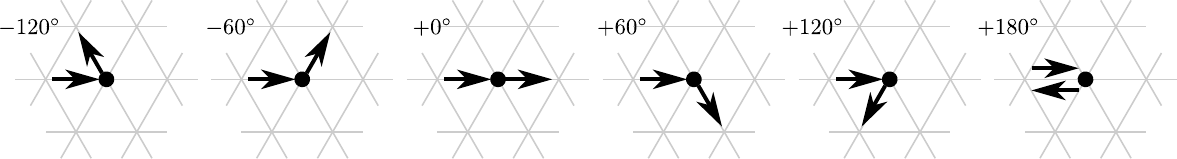}
\caption{Determining the external angle $\alpha$. The incoming and outgoing arrows represent the directions in which the token enters and leaves an agent, respectively. Only the angle between the arrows is relevant; the absolute global direction of the arrows cannot be detected by the agents since they do not posses a common compass.}
\label{fig:boundid}
\end{figure}

\section{Analysis} \label{sec:analysis}

We now turn to the analysis of the leader election algorithm.
We show its correctness in Section~\ref{subsec:correctness} and analyze its running time in Section~\ref{subsec:runningtime}.

\subsection{Correctness} \label{subsec:correctness}

To show the correctness of the algorithm we must prove that eventually a single particle irreversibly declares itself to be the leader of the particle system and no other particle ever declares itself to be the leader.
Any agent on an inner boundary can never cause its particle to become the leader; even if the algorithm reaches the point at which there is exactly one candidate $c$ on some inner boundary, $c$ will withdraw its candidacy in the boundary identification phase.
Therefore, we can focus exclusively on the behavior of the algorithm on the unique outer boundary.

We first show a sequence of lemmas to establish that with high probability there is a unique candidate that has an identifier that is strictly greater than the identifier of every other candidate.
Let $n$ be the number of particles in the particle system and let $L$ be the length of the outer boundary, i.e., the number of agents in the cycle spanning the outer boundary.
We have the following lower bound on $L$.

\begin{lemma} \label{lem:outerboundlen}
$L \ge \sqrt{n}$.
\end{lemma}
\begin{proof}
Define a coordinate system by picking an arbitrary node of $\Geqt$ as the origin and orienting the $x$ and $y$ axes as depicted in Figure~\ref{fig:coordboundbox}.
Consider the \emph{bounding box} of the particle system, i.e., the minimal parallelogram containing all particles (see Figure~\ref{fig:coordboundbox}).
We define the length of a side of this bounding box as the number nodes it spans.
Since the particle system is connected and contains $n$ particles, one of the sides of the bounding box has to be of length at least $\sqrt{n}$.
Suppose that the top and bottom sides of the bounding box are at least as long as the left and right sides (the opposite case is analogous).
Define a \emph{column} of the particle system to be a maximal set of particles with the same $x$-coordinate.
The particle with the greatest $y$-coordinate in each column is part of the outer boundary.
Since the number of columns equals the length of the top side, there are at least $\sqrt{n}$ agents on the outer boundary.
\qed
\end{proof}

The next lemma gives us a lower bound on the length of the longest segment which is equal to the number of digits in the longest identifier.

\begin{lemma} \label{lem:seglen}
For $n$ sufficiently large, there is a segment of length at least $0.25 \log n$, w.h.p.
\end{lemma}
\begin{proof}
Consider any agent $a_1$ on the outer boundary and let $a = (a_1, a_2, \ldots, a_L)$, where $a_i = a_{i-1}.\succ$ for $2 \le i \le L$.
Let $b = (b_1, b_2, \ldots, b_L)$, where $b_i = 1$ if the coin flip of $a_i$ came up heads and $b_i = 0$ otherwise.

We first show that $b$ contains a subsequence of length at least $m = 0.25 \log n$ such that all elements of the subsequence are $0$.
Divide $b$ into consecutive subsequences of length $m$.
By Lemma~\ref{lem:outerboundlen}, we get $k = \lfloor L / m \rfloor \ge \lfloor \sqrt{n} / m \rfloor$ such subsequences.
Since the case $n = 1$ is explicitly handled in Section~\ref{subsec:boundsetup}, we can assume $n \ge 2$, implying that $m \neq 0$ and $k$ is well-defined.
Let $b^{(i)} = (b_{(i - 1) m + 1}, \ldots, b_{im})$ be the $i$-th such subsequence where $1 \le i \le k$.
Let $E_i$ be the event that all elements of $b^{(i)}$ are $0$.
We have $\Pr[E_i] = (1 / 2)^m = n^{-1/4}$.
Since the events $E_i$ are independent, we have:
\[\Pr\left[\bigcap_{i = 1}^k \overline{E_i} \right]
  = \left(1 - n^{-1/4}\right)^k
  \le \left(1 - n^{-1/4}\right)^{\lfloor 4\sqrt{n} / \log n \rfloor}\]
Note that $\lfloor x \rfloor > x / 4$ for any $x > 1$.
Since $4 \sqrt{n} / \log n > 1$ for $n \ge 2$, this implies $\lfloor 4\sqrt{n} / \log n \rfloor > \sqrt{n} / \log n$. Therefore, we have:
\[\Pr\left[\bigcap_{i = 1}^k \overline{E_i} \right]
  < \left(1 - n^{-1/4}\right)^{\sqrt{n} / \log n}.\]
By rewriting the exponent and applying the well-known inequality $(1 - 1 / x)^x \le 1/e$, we get:
\[\Pr\left[\bigcap_{i = 1}^k \overline{E_i} \right]
  < \left(\left(1 - n^{-1/4} \right)^{n^{1/4}}\right)^{n^{1/4} / \log n}
  \le e^{-n^{1/4}/\log n}.\]
Therefore, we have that for $n$ sufficiently large the entries of one of the subsequences $b^{(i)}$ are all $0$, w.h.p.
This implies that $a$ contains at least $m$ consecutive non-candidates.

It remains to show that at least one element of $b$ is $1$ and therefore there is at least one candidate on the outer boundary.
By Lemma~\ref{lem:outerboundlen}, this holds with probability $1 - 2^{-L} \le 1 - 2^{-\sqrt{n}}$.
Therefore, for $n$ sufficiently large, there is a segment of size at least $0.25 \log n$, w.h.p.
\qed
\end{proof}

We can now show that there is a unique candidate that has an identifier that is strictly greater than the identifier of every other candidate.

\begin{lemma} \label{lem:highestid}
For $n$ sufficiently large, there is a candidate $c^*$ such that $c^*.\id > c.\id$ for every candidate $c \ne c^*$, w.h.p.
\end{lemma}
\begin{proof}
Let $C$ be the set of all candidates on the outer boundary, let $M \subseteq C$ be the subset of candidates with maximal segment length, and let $c^*$ be some candidate with the highest identifier.
Since shorter identifiers are defined to be less than longer identifiers, we must have $c^* \in M$ and $c^*.\id > c.\id$ for all $c \in C \setminus M$.
It remains to show that $c^*.\id > c.\id$ for all $c \in M \setminus \{ c^* \}$.
This is the case if the identifier of $c^*$ is unique.

By definition, the identifiers of the candidates in $M$ all consist of the same number of digits, which must be at least $0.25 \log n$ by Lemma~\ref{lem:seglen}.
Each digit is chosen independently and uniformly at random from the interval $[0, r - 1]$ for a constant $r$ of our choice.
Therefore, for any candidate $c \in M \setminus \{ c^* \}$ the probability that $c.\id = c^*.\id$ is at most $r^{-0.25 \log n}$.
Applying the union bound gives us that the probability that there exists such a candidate is at most:
\[(|M| - 1) \cdot r^{-0.25 \log n} < n^{1 - 0.25 \log r},\]
which concludes the proof.
\qed
\end{proof}

We can now prove the correctness of the algorithm.

\begin{theorem} \label{thm:correctness}
The algorithm solves the leader election problem, w.h.p.
\end{theorem}
\begin{proof}
We must show that eventually a single particle irreversibly declares itself to be the leader of the particle system and no other particle ever declares itself to be the leader.
Again, we consider only the agents on the outer boundary as agents on an inner boundary will never cause their particles to declare themselves as leaders.
Once every particle has finished the boundary setup phase, every agent has finished the segment setup phase, and every candidate has finished the identifier setup phase, with high probability there is a unique candidate $c^*$ that has the highest identifier on the outer boundary by Lemma~\ref{lem:highestid}.
Since $c^*$ has the highest identifier, it does not withdraw its candidacy during the identifier comparison phase.
In contrast, every other candidate $c \ne c^*$ eventually withdraws its candidacy because the token sequence of $c^*$ eventually traverses $c.\seg$.
Therefore, such an agent $c$ cannot cause its particle to become the leader.
Once $c^*$ is the last remaining candidate on the outer boundary, it eventually triggers the solitude verification phase because the token sequence of $c^*$ eventually traverses $c^*.\seg$ while $c^*$ is not already performing solitude verification.
After verifying that it is the last remaining candidate, $c^*$ executes the boundary identification phase and determines that it lies on the outer boundary.
It then instructs its particle to declare itself the leader of the particle system.
\qed
\end{proof}

\subsection{Running Time} \label{subsec:runningtime}

Recall from Section~\ref{subsec:problem} that the running time of an algorithm for leader election is defined as the number of asynchronous rounds until a leader is declared.
Since the given algorithm always establishes a leader on the outer boundary, we can limit our attention to that boundary.
The first two phases of the algorithm, namely the boundary setup and segment setup phases, consist entirely of computations based on local neighborhood information.
Therefore, these phases can be completed instantly by each particle upon its first activation.
Since each particle is activated at least once in every round, every particle completes these first two phases after a single round.

When an agent becomes a candidate, it initiates the identifier setup phase.
We have the following lemma.

\begin{lemma} \label{lem:segidtime}
For a segment of length $\ell$, the identifier setup takes $\mathcal{O}(\ell^2)$ rounds.
\end{lemma}
\begin{proof}
In the identifier setup phase, a token that is created by the candidate first traverses the segment to establish the random digits of the identifier.
Once the token reaches the end of the segment, it creates a copy of the identifier that is stored in reversed order in the segment by moving in an alternating fashion back and forth through the segment.
Finally, once the copying is complete, the token moves back to the candidate.

The token must be forwarded at least once per round since there are no other tokens blocking it; thus, we can upper bound the number of rounds required by the length of the token's trajectory. 
It is not difficult to see that the token takes $\mathcal{O}(\ell^2)$ steps, so we conclude the identifier setup phase takes $\mathcal{O}(\ell^2)$ rounds.
\qed
\end{proof}

To bound the number of rounds required to complete the identifier setup phase for all segments on the outer boundary, we have to bound the maximal length of a segment.

\begin{lemma} \label{lem:seglentime}
The length of a segment on the outer boundary is $\mathcal{O}(\log n)$, w.h.p.
\end{lemma}
\begin{proof}
For any constant $k \in \mathbb{R}^+$, the probability that an agent becomes a candidate with a segment of length at least $k \log n$ is at most $(1 / 2)^{k \log n} = n^{-k}$.
Since there are $n$ particles in the particle system and each particle corresponds to at most $3$ agents, there are at most $3n$ agents on the outer boundary.
Applying the union bound shows that the probability of there being a segment of length at least $k \log n$ is at most $3n^{1-k}$, which proves the lemma.
\qed
\end{proof}

Combining the previous two lemmas gives us the following corollary.

\begin{corollary} \label{cor:idsetuptime}
All candidates on the outer boundary complete the identifier setup phase after $\mathcal{O}(\log^2 n)$ rounds, w.h.p.
\end{corollary}

After the identifiers have been generated, they are compared in the identifier comparison phase.
In this phase, a set of digit tokens, one for each agent on the boundary, traverses the boundary against the direction of the cycle spanning it.
Each agent can store at most two tokens.
The tokens are not allowed to overtake each other, so agents maintain the order of the tokens when forwarding them.
Note that a token is never delayed unless it is blocked by tokens in front of it.
Therefore, an agent $a$ forwards a token whenever $a.\pred$ can hold an additional token.
Finally, an agent forwards at most one token for each activation.

We define the number of \emph{steps} a token has taken as the number of times it's been forwarded from one agent to the next since its creation.
Let $T$ be the earliest round such that at its beginning every agent on the outer boundary has created its digit token.
We have the following lemma.

\begin{lemma} \label{lem:idcomptime}
At the beginning of round $T + i$ for $i \in \mathbb{N}$, each digit token on the outer boundary has taken at least $i$ steps.
\end{lemma}
\begin{proof}
We establish a lower bound on the number of steps a token took by comparing the \emph{asynchronous} execution of the token passing scheme with a \emph{synchronous} execution in which tokens move in lockstep.
For the synchronous execution, we assume that each token is initially stored at the agent that created it.
We refer to this point in time in the synchronous execution as round $0$.
The tokens move in lockstep along the boundary so that in every round each agent stores exactly one token.
For a token $t$ let $s_i(t)$ be the number of steps $t$ has taken by the beginning of round $i$ of the synchronous execution; by definition, we have $s_i(t) = i$.
Similarly, let $a_i(t)$ be the number of steps $t$ has taken at the beginning of round $T + i$ of the asynchronous execution.
We show by induction on $i$ that $a_i(t) \ge s_i(t)$ for all tokens $t$.

The statement holds for $i = 0$ by definition.
Now suppose that the statement holds for some $i \ge 0$ and consider any token $t$.
We show $a_{i+1}(t) \ge s_{i+1}(t)$.
If $a_i(t) > s_i(t)$ then
\[a_{i+1}(t) \ge a_i(t) \ge s_i(t) + 1 = s_{i+1}(t),\]
and therefore the statement holds.
So assume $a_i(t) = s_i(t)$.
To prove $a_{i+1}(t) \ge s_{i+1}(t)$, we have to show that $t$ is forwarded at least once in round $T + i$ of the asynchronous execution.
Since the tokens do not overtake each other, their relative order along the boundary (considering both a token's position on the boundary as well as which token is forwarded first if an agent holds two tokens) is well-defined and remains unchanged in both the synchronous and the asynchronous execution.
Let $t'$ be the token in front of $t$ in the direction of the traversal (i.e., against the direction of the cycle spanning the boundary).
In the synchronous execution, $t'$ started one agent ahead of $t$.
Furthermore, the induction hypothesis together with our assumption that $a_i(t) = s_i(t)$ implies:
\[a_i(t') \ge s_i(t') = s_i(t) = a_i(t).\]
Therefore, $t'$ is at least one agent ahead of $t$ in the asynchronous execution.
By an analogous argument, one can see that the token $t''$ following $t'$ is at least two agents ahead of $t$ in the asynchronous execution.
Since the tokens preserve their order, there are no other tokens in between $t$, $t'$, and $t''$.
Thus, the agent to which $t$ should be forwarded in round $T + i$ of the asynchronous execution holds at most one token, namely $t'$.
Furthermore, even if the agent holding $t$ at the beginning of round $T + i$ of the asynchronous execution holds a second token, $t$ is forwarded first because the order of the tokens is preserved.
Therefore, once the agent holding $t$ is activated in round $T + i$ of the asynchronous execution --- which must occur since every particle (and therefore agent) is activated at least once per round --- $t$ is indeed forwarded.
\qed
\end{proof}

In the proof of Lemma~\ref{lem:idcomptime} we showed that the asynchronous execution of a token passing scheme \emph{dominates} a corresponding synchronous execution in terms of a suitable measure of progress. Next, we analyze the runtime of the solitude verification phase using a similar ``dominating argument''.

\begin{lemma} \label{lem:solitudevertime}
For an extended segment of length $\ell$, the solitude verification phase takes $\mathcal{O}(\ell)$ rounds.
\end{lemma}
\begin{proof}
The token passing scheme of the solitude verification phase is executed independently for the $x$ and the $y$ axes.
We consider one of these executions and show that it takes $\mathcal{O}(\ell)$ rounds.
First, the activation token moves through the extended segment and creates positive and negative tokens.
Since the activation token moves through the extended segment unhindered, this traversal takes $\mathcal{O}(\ell)$ rounds.
It then moves back towards the candidate, but remains behind the positive and negative tokens it created.
These two types of tokens move back towards the candidate independently of each other.
However, two tokens of the same type cannot overtake each other.
Once all tokens of both types have settled, the activation token can move back to the candidate unhindered which takes another $\mathcal{O}(\ell)$ rounds.
Thus, it remains to determine the number of rounds until all tokens have settled.

We use a dominating argument similar to that of Lemma~\ref{lem:idcomptime}.
Consider a single token type; w.l.o.g., suppose this is positive.
To simplify our notation, we define round $0$ of the asynchronous execution to be the earliest round such that at its beginning the activation token has already created all positive tokens.
We compare the asynchronous execution with the following synchronous execution: in round $0$ of the synchronous execution, each token is stored at the agent that created it.
The tokens then move in lockstep towards the candidate.
Apart from the movement of the tokens, the synchronous execution works the same way as the asynchronous execution, i.e., each agent can store two tokens and the tokens move as close to the candidate as possible.
We assign the numbers $1, \ldots, \ell$ to the agents of the extended segment starting with $1$ at the candidate.
For a token $t$ let $s_i(t)$ (respectively, $a_i(t)$) be the number assigned to the agent that holds $t$ at the beginning of round $i$ of the synchronous (respectively, asynchronous) execution.
We show by induction on $i$ that $a_i(t) \le s_i(t)$ for all tokens $t$.

The statement holds for $i = 0$ by definition.
Suppose that the statement holds for some round $i \ge 0$ and consider a token $t$.
We show $a_{i+1}(t) \le s_{i+1}(t)$.
If $a_i(t) < s_i(t)$ then:
\[a_{i+1}(t) \le a_i(t) \le s_i(t) - 1 \le s_{i+1}(t),\]
and therefore the statement holds.
So assume $a_i(t) = s_i(t)$.
We need two observations: first, since the tokens do not overtake each other, their order along the extended segment is well-defined and remains unchanged in both the synchronous and the asynchronous execution; second, when an agent holds two tokens in the synchronous execution, both tokens must be settled at their final position.
Let $t'$ be the next token from $t$ towards the candidate.
If there is no such token, the statement holds since $t$ can move unhindered in the asynchronous execution.
Otherwise, we have $s_i(t') \le s_i(t)$ by our first observation.
We distinguish three cases and show that in each case $a_{i+1}(t) \le s_{i+1}(t)$.

\begin{enumerate}
\item If $s_i(t') = s_i(t)$, then by our second observation both $t$ and $t'$ are at their final position. Therefore, $t$ is never forward again in both the synchronous and asynchronous executions. This implies:
\[a_{i+1}(t) = a_{i}(t) = s_i(t) = s_{i+1}(t).\]

\item If $s_i(t') \le s_i(t) - 2$ then the agent ahead of $t$ holds no token in the synchronous execution because there is no token between $t$ and $t'$ by our first observation. By our assumption that $a_i(t) = s_i(t)$ together with the induction hypothesis applied to $t'$, the agent ahead of $t$ also holds no token in the asynchronous execution. Therefore, $t$ is forwarded at least once in the asynchronous execution and we have:
\[a_{i+1}(t) \le a_{i}(t) - 1 = s_i(t) - 1 = s_{i+1}(t).\]

\item If $s_i(t') = s_i(t) - 1$ then we must distinguish between two subcases. If agent $s_i(t')$ holds two tokens in the synchronous execution, then $t'$ is at its final position and hence $t$ is also at its final position. Therefore, the statement holds by the same argument as in the first case above.

So suppose that $s_i(t')$ only holds $t'$ in the synchronous execution.
Let $t''$ be the next token from $t'$ towards the candidate.
If there is no such token, then the agent ahead of $t$ holds at most one token in the asynchronous execution, namely $t'$.
Otherwise, $t''$ is at least two agents ahead of $t$ in the synchronous execution and, by the induction hypothesis, also in the asynchronous execution.
So again, the agent ahead of $t$ holds either no token or only $t'$ in the asynchronous execution.
Therefore, $t$ is forwarded at least once in the asynchronous execution and the statement holds by the same argument as in the second case above.
\end{enumerate}

It is easily seen that the synchronous execution moves all tokens to their final positions in $\mathcal{O}(\ell)$ rounds.
Therefore, the asynchronous execution requires at most $\mathcal{O}(\ell)$ rounds to do the same.
Once all tokens reach their final position, it takes $\mathcal{O}(\ell)$ additional rounds until each token sets the bit showing that it is settled to true.
\qed
\end{proof}

The boundary identification phase is only executed once a candidate determines that it is the last remaining candidate on the boundary.
The following lemma provides an upper bound for the running time of this phase.
Recall that $L$ is defined as the number of agents on the outer boundary.

\begin{lemma} \label{lem:boundidtime}
The boundary identification phase on the outer boundary takes $\mathcal{O}(L)$ rounds.
\end{lemma}
\begin{proof}
The token calculating the angle completely traverses the outer boundary.
Since the boundary has length $L$ and the token is forwarded at least once in every round, this takes $\mathcal{O}(L)$ rounds.
\end{proof}

Finally, we can show the following runtime bound.

\begin{theorem} \label{thm:runningtime}
The algorithm solves the leader election problem in $\mathcal{O}(L)$ rounds, w.h.p.
\end{theorem}
\begin{proof}
After the first round, the boundaries, candidates, and segments have been established.
By Corollary~\ref{cor:idsetuptime}, all candidates on the outer boundary complete the identifier setup phase within $\mathcal{}(\log^2 n)$ rounds.
Thus, according to Lemma~\ref{lem:highestid}, a unique candidate $c^*$ with the highest identifier on the outer boundary has been established, w.h.p.
We show that $c^*$ declares itself to be the leader after $\mathcal{O}(L)$ rounds.
At the beginning of round $T = \mathcal{O}(\log^2 n)$, all digit tokens have been created.
According to Lemma~\ref{lem:idcomptime}, after an additional $L$ rounds every digit token has completed a pass along the outer boundary.
At this point, the token sequence of $c^*$ has traversed the segment of every other candidate and, therefore, every candidate except for $c^*$ either has already withdrawn its candidacy or is flagged to do so after completing its current execution of the solitude verification phase.
Since the maximum length of an extended segment of a candidate on the outer boundary is $L$, the candidates that still execute the solitude verification phase withdraw their candidacy after $\mathcal{O}(L)$ additional rounds by Lemma~\ref{lem:solitudevertime}.
Once $c^*$ is the only remaining candidate on the outer boundary, it has to start its final execution of the solitude verification phase.
However, it might have to finish an already running execution of the solitude verification phase which could fail because it was started too early.
In this case, it takes $\mathcal{O}(L)$ rounds for the solitude verification phase to fail and an additional $L$ rounds for the token sequence of $c^*$ to traverse $c^*.\seg$ again, thus triggering the final execution of the solitude verification phase.
This takes another $\mathcal{O}(L)$ rounds.
When $c^*$ determines that it is the last remaining candidate, it executes the boundary identification phase, which takes $\mathcal{O}(L)$ rounds according to Lemma~\ref{lem:boundidtime}.
So after a total of $\mathcal{O}(L)$ rounds, $c^*$ determines that it is the only remaining candidate on the outer boundary and signals its particle to declare itself the leader.
\qed
\end{proof}

Theorem~\ref{thm:runningtime} specifies the running time of the leader election algorithm in terms of the number of \emph{agents} on the outer boundary.
Let $C$ be the number of \emph{particles} on the outer boundary.
Since each particle on the outer boundary corresponds to at most three agents on the outer boundary, we have the following corollary.

\begin{corollary} \label{cor:runningtimeC}
The algorithm solves the leader election problem in $\mathcal{O}(C)$ rounds, w.h.p.
\end{corollary}

Depending on the application it might be desirable to specify the running time of the algorithm in terms of the number of particles $n$ in the entire particle system.
Clearly, the number of particles on the outer boundary is at most $n$, implying the following corollary.

\begin{corollary} \label{cor:runningtimeN}
The algorithm solves the leader election problem in $\mathcal{O}(n)$ rounds, w.h.p.
\end{corollary}

Note that compared to Corollary~\ref{cor:runningtimeC}, Corollary~\ref{cor:runningtimeN} represents a quite pessimistic bound on the running time of the algorithm since the number of particles on the outer boundary can much lower than $n$.
For example, a solid square of $n$ particles (formed like the bounding box depicted in Figure~\ref{fig:coordboundbox}) only has $C = \mathcal{O}(\sqrt n)$ particles on its outer boundary.

\section{Conclusion} \label{sec:conclude}

In this paper we presented a randomized leader election algorithm for programmable matter which requires $\mathcal{O}(n)$ asynchronous rounds with high probability.
The main idea of this algorithm is to use coin flips to set up random identifiers for each leader candidate in such a way that at least one candidate has an identifier of logarithmic length.
We then use a token passing scheme that allows us to compare all identifiers on one boundary with each other in a pipelined fashion, respecting the constant-size memory constraints at each particle.
We exploit the geometric properties of the triangular grid only in the last two phases of our algorithm, in which a candidate tests whether it is the last candidate on a boundary and whether it is on the outer boundary.
If a candidate passes both of these tests, it irreversibly declares itself the unique leader.

\bibliographystyle{abbrv}
\bibliography{literature}

\appendix

\section{Variants of the Leader Election Problem} \label{app:variants}

In this section, we consider variants of the leader election problem.
We present three positive results: (1) how a leader can be elected when the particle system contains expanded particles, (2) how the leader election algorithm can be extended such that its execution terminates for all particles, and (3) how to improve the with high probability guarantee on the election of a leader to an almost surely guarantee without changing the $\mathcal{O}(L)$ runtime.
The three variant algorithms can be combined into a single algorithm that satisfies all of the above properties.
We close this section with a negative result concerning the generalization of the leader election problem to arbitrary graphs.

\subsection{Expanded Particles} \label{subapp:expanded}

It is straightforward to extend the leader election algorithm to allow particle systems containing expanded particles.
An expanded particle $p$ simply simulates two distinct contracted particles, one for each node occupied by $p$.
Whenever $p$ is activated, it simulates the activations of the corresponding contracted particles one after another in an arbitrary order, effectively reducing the problem of leader election with expanded particles to leader election without expanded particles.

Since every physical particle is activated at least once in every round, also every simulated particle is activated at least once in every round.
Furthermore, the running time analysis presented in Section~\ref{subsec:runningtime} holds for arbitrary activation orders.
Therefore, the analysis remains valid despite the fact that two simulated particles of an expanded physical particle are always activated immediately after one another.
This implies that the statements of Theorem~\ref{thm:runningtime}, Corollary~\ref{cor:runningtimeC}, and Corollary~\ref{cor:runningtimeN} also hold for particle systems containing expanded particles.

\subsection{Termination for All Particles} \label{subapp:terminationforall}

In the definition of the leader election problem given in Section~\ref{subsec:problem}, the leader is the only particle for which the algorithm must terminate, while any non-leader particle is allowed to execute the algorithm indefinitely.
Accordingly, the algorithm presented in Section~\ref{sec:algo} can actually experience infinite loops for a subset of the agents in certain situations.
For example, consider a particle system with an empty region $R$ of size $1$.
With constant probability, all six agents on the inner boundary corresponding to $R$ become candidates and get the same one-digit identifier.
The identifier comparison phase of the leader election algorithm will never eliminate any of the six candidates in this situation and thus the candidates remain stuck in an infinite loop.

Depending on the application, it might be desirable to have a leader election algorithm that is guaranteed to terminate for all particles.
This can be achieved using the following extension of the algorithm presented in Section~\ref{sec:algo}.
After the leader has been elected, it broadcasts a \emph{termination message} through the particle system.
A particle receiving this message forwards it to each of its neighbors by writing it into their memories and then terminates its execution of the leader election algorithm.
Let $A$ be the set of occupied nodes in $\Geqt$ and let $\Geqt|_A$ be the subgraph of $\Geqt$ induced by $A$.
The running time of this termination broadcast is linear in the diameter $D$ of $\Geqt|_A$, so after $\mathcal{O}(L + D)$ rounds, a leader has emerged with high probability and the execution of the leader election algorithm has terminated for all particles in the particle system.
We summarize this result in the following corollary.

\begin{corollary} \label{cor:allterminate}
There is a leader election algorithm that terminates for all particles in $\mathcal{O}(L + D)$ rounds, w.h.p.
\end{corollary}

Note that the parameters $L$ and $D$ are, in general, independent of each other.
For example, it is not hard to construct particle systems such that either $L = \Omega(n)$ and $D = O(\sqrt n)$ or $L = O(\sqrt n)$ and $D = \Omega(n)$.
Nevertheless, $L$ and $D$ are both clearly upper bounded by $n$.
This implies the following corollary.

\begin{corollary} \label{cor:allterminateN}
There is a leader election algorithm that terminates for all particles in $\mathcal{O}(n)$ rounds, w.h.p.
\end{corollary}

\subsection{Almost Sure Leader Election} \label{subapp:almostsurele}

The leader election algorithm presented in Section~\ref{sec:algo} elects a leader with high probability.
Accordingly, there is a small but nonzero probability that the algorithm fails to elect a leader if either every agent becomes a non-candidate during the segment setup phase or more than one candidate generates the same highest identifier in the identifier setup phase.
In this section, we describe how the leader election algorithm can be extended such that it \emph{almost surely} (i.e., with probability 1) elects a leader while maintaining the runtime bound of $\mathcal{O}(L)$ rounds with high probability.

The main idea of this extension is to run a second leader election algorithm in parallel to the algorithm presented in Section~\ref{sec:algo}.
The second algorithm sets up the boundaries as described in Section~\ref{subsec:boundsetup}.
Each agent is initially a candidate, and the candidates alternate between the following two phases.
In the first phase, a candidate flips a coin and sends the result along its boundary to both its preceding and its succeeding candidate.
A candidate withdraws its candidacy if it flips tails while both its predecessor and successor flip heads.
Note that this competition locally synchronizes competing candidates.
The second phase of the algorithm corresponds to the solitude verification phase described in Section~\ref{subsec:solitudever}.

Once a candidate determines that it is the last remaining candidate on its boundary, it executes the boundary identification phase described in Section~\ref{subsec:boundid}.
If the candidate lies on an inner boundary, it withdraws its candidacy.
If it lies on the outer boundary, it sends a token along the boundary that stops the execution of the original algorithm in the particles on the outer boundary and, at the same time, checks whether the original algorithm already established a leader.
If there already is a leader, the candidate withdraws its candidacy.
Otherwise, it declares itself the leader.
We have the following theorem.

\begin{theorem} \label{thm:almostsure}
The algorithm elects a leader in $\mathcal{O}(L)$ rounds w.h.p., and eventually solves the leader election problem almost surely.
\end{theorem}
\begin{proof}
Consider the execution of the second algorithm on the outer boundary.
As long as there is more than one candidate on the boundary, the first phase of the algorithm reduces the number of candidates with a probability that is lower bounded by a constant.
Furthermore, the last remaining candidate on the boundary competes with itself and will therefore never withdraw its candidacy during the first phase of the algorithm.
Thus, it holds almost surely that eventually only a single candidate remains on the outer boundary.

The last remaining candidate on the outer boundary eventually passes the solitude verification phase and the boundary identification phase.
It then interacts with the original leader election algorithm by sending a token along the outer boundary, which produces one of three results: (1) if the original algorithm establishes a leader, the second algorithm does not produce a leader; (2) if the original algorithm establishes a unique candidate with the highest identifier on the outer boundary but the particle corresponding to that candidate is reached by the aforementioned token before it claims leadership, only the second algorithm establishes a leader; and (3) if the original algorithm fails to establish a unique candidate with the highest identifier on the outer boundary, the second algorithm establishes a leader.

With high probability either the first or second case holds according to Theorem~\ref{thm:correctness}.
In the first case, the leader is established by the original algorithm in $\mathcal{O}(L)$ rounds by Theorem~\ref{thm:runningtime}.
In the second case, the second algorithm stops the execution of the original algorithm before it establishes a leader.
Since the original algorithm establishes a leader in $\mathcal{O}(L)$ rounds, the aforementioned token must have been created in $\mathcal{O}(L)$ rounds.
The token requires at most $L$ rounds to traverse the boundary and, therefore, the second algorithm establishes a leader in $\mathcal{O}(L)$ rounds.
Finally, in the third case the second algorithm almost surely establishes a leader eventually.
\qed
\end{proof}

\subsection{General Graphs} \label{subapp:generalgraphs}

In the amoebot model, the particles occupy the nodes of the infinite triangular grid graph $\Geqt$.
The graph $\Geqt$ is embedded in the Euclidean plane, which provides the particles with geometric information.
Specifically, the particles can measure distance because each pair of nodes connected by an edge is separated by a unit distance, and they can measure angles since they know the rotational order of neighboring nodes through their port labels and each face in the graph is an equilateral triangle.
The leader election algorithm presented in Section~\ref{sec:algo} explicitly uses this geometric information in the boundary setup, solitude verification, and boundary identification phases.
A natural question is whether there is an algorithm that achieves leader election for an arbitrary graph $G$ and without geometric information.
To formally investigate this question, we assume that we are given a set of contracted particles occupying a connected subset of the nodes of a graph $G$ of constant degree $d$ and the ports of the particles are labeled arbitrarily with numbers from $1$ to $d$.\footnote{This variant of the model was presented in~\cite{DNA} as the \emph{general amoebot model} while the model presented in this work was referred to as the \emph{geometric amoebot model}.}

Let $G$ be a ring of $n$ nodes and let each node be occupied by a contracted particle.
As before, to solve the leader election algorithm a single particle must irreversibly declare itself the leader and no other particle may ever declare itself to be the leader.
Since every node is occupied by a contracted particle, the particles cannot move.
Thus, the particles have to elect a leader using communication only, without geometric information.
Therefore, the particle system forms an undirected ring of anonymous nodes where each node is a probabilistic finite automaton.

We say a leader election algorithm fails if it either does not establish a leader or it establishes more that one leader.
The problem of electing a leader in a ring of anonymous nodes has already been investigated by Itai and Rodeh~\cite{IR90}.
Their results imply that for every $\rho < 1$, the failure probability of an algorithm for leader election on rings of arbitrary size is greater than $\rho$, i.e., the error probability cannot be bounded away from $1$.
Therefore, the leader election problem is infeasible on arbitrary graphs without geometric information.

\end{document}